\documentclass{IEEEtran}

\usepackage{amsmath, amsfonts, amssymb, graphicx}
\usepackage{pstricks,pst-node}
\usepackage{cite}

\newtheorem{theorem}{Theorem}
\newtheorem{lemma}[theorem]{Lemma}
\newtheorem{prop}[theorem]{Proposition}
\newtheorem{corollary}[theorem]{Corollary}
\renewcommand{\vec}[1]{\mathbf{#1}}
\def\ev{\mathop{\rm eval}\nolimits}

\long\def\symbolfootnote[#1]#2{\begingroup
\def\thefootnote{\fnsymbol{footnote}}\footnote[#1]{#2}\endgroup}

\newcounter{example}
\newenvironment{example}[1][]{\refstepcounter{example}\par\medskip\noindent%
   \textbf{Example~\theexample. }}{\medskip}

\newcounter{remark}
\newenvironment{remark}[1][]{\refstepcounter{remark}\par\medskip\noindent%
   \textbf{Remark~\theremark. }}{\medskip}

\title{Achieving Capacity of Bi-Directional Tandem Collision Network by Joint Medium-Access Control and Channel-Network Coding}

\author{Kenneth W. Shum and Chi Wan Sung}

\begin{document}

\maketitle

\begin{abstract}
In ALOHA-type packetized network, the transmission times of packets follow a stochastic process. In this paper, we advocate a deterministic approach for channel multiple-access. Each user is statically assigned a periodic protocol signal, which takes value either zero or one, and transmit packets whenever the value of the protocol signal is equal to one.  On top of this multiple-access protocol, efficient channel coding and network coding schemes are devised. We illustrate the idea by constructing a transmission scheme for the tandem collision network, for both slot-synchronous and slot-asynchronous systems. This cross-layer approach is able to achieve the capacity region when the network is bi-directional.
\end{abstract}

\begin{keywords} collision channel, protocol sequence, tandem network, bi-directional network, network coding.
\end{keywords}

\section{Introduction}
\symbolfootnote[0]{This work was partially supported by a grant from the Research Grants Council of the Hong Kong Special Administrative Region under Project 417909, and a grant from City University of Hong Kong under project 7002386.}
\symbolfootnote[0]{Kenneth Shum is with Dept. of Information Engineering, the Chinese University of Hong Kong, Shatin, Hong Kong.}
\symbolfootnote[0]{Chi Wan Sung is with Dept. of Electronic Engineering, City University of Hong Kong, Tat Chee Ave, Kowloon Tong, Hong Kong.}
\symbolfootnote[0]{Emails: kshum2009@gmail.com, albert.sung@cityu.edu.hk.}

In their study of  multiple-access collision channel without feedback, Massey and Mathys show  that the capacity region can be achieved by deterministic channel access method~\cite{MasseyMathys85}. It contrasts with the more traditional multiple-access scheme like pure ALOHA or slotted ALOHA\cite{Abramson70, Abramson73}, where transmission times of packets form a random process. In this paper, we extend the transmission scheme by Massey and Mathys to tandem network, in which nodes are located on a straight line. Nodes that are more than two hops away do not interfere with each other.
If a node receives two packets that overlap in time, either partially or completely, both packets are assumed erased and unrecoverable.
This model is applicable to wireless sensor network along a highway or river for instance. Since antenna system that transmits and receives at the same time is often too costly to implement, we assume that each node operates in half-duplex mode; when a node is not transmitting, it listens to its two neighboring nodes.

In our proposed transmission scheme, the transmission times of packets of each node follow a pre-assigned deterministic and periodic pattern.
Implementation of this kind of deterministic channel accessing scheme is particularly easy. We can simply store the whole pattern of transmission times in memory, read it out repeatedly, and send out packets accordingly.  No coordination between nodes and no centralized packet scheduling is needed. Network operations are thus fully distributed. This feature is especially suitable for low-complexity wireless sensor network.


Related work on tandem network can be found in~\cite{Pakzad2005, Niesen2007, Sagduyu06, Sagduyu07}. In~\cite{Pakzad2005}, the authors consider one source node and one destination node, that are connected by a series of intermediate relay nodes in tandem, and model each hop as an erasure channel. A coding scheme which can approach the min-cut bound~\cite{ThomasCover} is proposed. Similar system setting with general discrete memoryless channel in each hop is considered in~\cite{Niesen2007}, and scaling laws for capacity is discussed. The networks considered in~\cite{Pakzad2005} and~\cite{Niesen2007} are full-duplex wireline network. In~\cite{Sagduyu06, Sagduyu07},
half-duplex wireless network is investigated, and the effect of
random multiple-access on network coding~\cite{ACLY00} is addressed.
In this paper, we devise a transmission scheme that combines  multiple-access protocol,  erasure-correction coding and network coding efficiently. We will show that this cross-layer design can achieve the (zero-error) capacity region of bi-directional tandem collision network, and is thus optimal.


This paper is organized as follows. The system model is introduced in Section~\ref{sec:system}. We consider both slot-synchronous and slot-asynchronous case.
A transmission scheme that incorporates multiple-access, erasure correction and network coding for slot-synchronous system is presented in Sections~\ref{sec:achievable}. Then we show that it can be extended to slot-asynchronous system. In order to show that the proposed transmission scheme is optimal for the bi-directional tandem collision network, we derive an outer bound on capacity region in Section~\ref{sec:converse}, and show that the achievable rate region and the outer bound coincide. Comparison with some random access schemes is carried out in Section~\ref{sec:example}. We close with some concluding remarks in Section~\ref{sec:conclusion}.

\section{System Model and Definitions} \label{sec:system}

We consider $M$ nodes located on a straight line. Each node broadcasts signal to its two neighbors, one on the left and one on the right.
It is assumed that the transmit range of each node is adjusted so that there is no interference to the nodes that are two or more hops away.  We model the network as a directed graph with vertex set $\mathcal{V} = \{1,2,\ldots, M\}$ and edge set
\begin{align*}
\mathcal{E} &= \{ (i,i+1):\, i=1,2,\ldots, M-1\} \\
& \quad \cup \{ (i+1,i):\, i=1,2,\ldots, M-1\}.
\end{align*}
We assume that all nodes operate in half-duplex mode, meaning that each of them cannot transmit and receive at the same time.

Suppose that there are $N$ independent data sources, and each data source is associated with a node. For $j=1,2,\ldots, N$, we let $\alpha(j)$ be the vertex to which the $j$-th source is attached. The function $\alpha:\{1,\ldots, N\} \rightarrow \mathcal{V}$ is called the {\em source mapping}. Each source is multicast  to a subset of nodes in~$\mathcal{V}$.
Define the {\em destination mapping}, $\beta$, which is a function from $\{1,2,\ldots, N \}$ to $2^\mathcal{V}$, such that for $j=1,2,\ldots, N$, source $j$ is demanded by all nodes belonging to set $\beta(j)$.
Node $i$ is called a {\em source node} if $\alpha(j) = i$ for some source~$j$.
If $i \in \beta(j)$ for some $j$, then node~$i$ is called a {\em destination node}. We remark that a node may be a source node and a destination node simultaneously.
A node that is neither a source nor a destination node is called a {\em pure relay node}.

\begin{example}
(Two-way network) The two nodes at the two ends send data to each other, and the nodes in the middle are pure relay nodes. Fig.~\ref{fig:linear_network1} illustrates an example for $M = 4$. There are two data sources, and so we have $N=2$. Nodes 1 wants to send message $W_1$ to node 4, and node 4 wants to send message $W_2$ to node~1. In this example, we have $\alpha(1) = 1$, $\alpha(2) = 4$, $\beta(1) = \{4\}$, $\beta(2) = \{1\}$.  The message in square bracket signifies that it is demanded by the associated node. Nodes 2 and 3 are pure relay nodes.
\end{example}

\begin{figure}
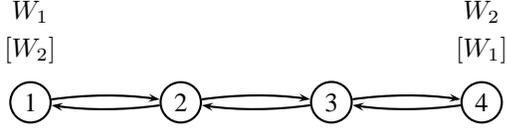

    \begin{center}
        \pspicture(-1,-.5)(7,2)
     \cnodeput(0,0){N1}{1}
     \cnodeput(2,0){N2}{2}
     \cnodeput(4,0){N3}{3}
     \cnodeput(6,0){N4}{4}

     \ncarc{->}{N1}{N2}
     \ncarc{->}{N2}{N1}
     \ncarc{->}{N2}{N3}
     \ncarc{->}{N3}{N2}
     \ncarc{->}{N3}{N4}
     \ncarc{->}{N4}{N3}
     \rput[u](0,0.7){$[W_2]$}
     \rput[u](0,1.2){$W_1$}
     \rput[u](6,0.7){$[W_1]$}
     \rput[u](6,1.2){$W_2$}
        \endpspicture
    \end{center}
\caption{A Two-way Network}
\label{fig:linear_network1}
\end{figure}

\begin{example}
(Bi-directional multicast network) The network consists of five nodes.
Two sources are associated with nodes 2 and~4. Node~2 wants to send message $W_1$ to nodes 1 and 5, and node 4 wants to send message $W_2$ to nodes 1 and~5 (Fig.~\ref{fig:linear_network2}).
The system parameters are $M=5$, $N=2$, $\alpha(1) = 2$, $\alpha(2) = 4$,  and $\beta(1) = \beta(2) = \{1,5\}$.  Node~3 acts as a pure relay node.

\begin{figure}
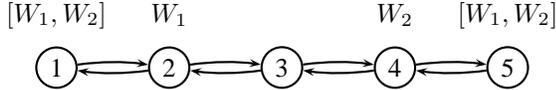

    \begin{center}
        \pspicture(-.5,-.5)(6.5,2)
     \cnodeput(0,0){N1}{1}
     \cnodeput(1.5,0){N2}{2}
     \cnodeput(3,0){N3}{3}
     \cnodeput(4.5,0){N4}{4}
     \cnodeput(6,0){N5}{5}

     \ncarc{->}{N1}{N2}
     \ncarc{->}{N2}{N1}
     \ncarc{->}{N2}{N3}
     \ncarc{->}{N3}{N2}
     \ncarc{->}{N3}{N4}
     \ncarc{->}{N4}{N3}
     \ncarc{->}{N4}{N5}
     \ncarc{->}{N5}{N4}
     \rput[u](0,0.7){$[W_1, W_2]$}
     \rput[u](6,0.7){$[W_1, W_2]$}
     \rput[u](1.5,0.7){$W_1$}
     \rput[u](4.5,0.7){$W_2$}
        \endpspicture
    \end{center}
\caption{Bi-directional Multicast  Network}
\label{fig:linear_network2}
\end{figure}
\end{example}

We assume that the data stream is packetized, and the durations of all packets are $T$ seconds. Each packet may assume $Q$ possible values, for some positive integer~$Q$. Thus, each packet carries $\lfloor \log_2(Q) \rfloor $ bits.
Consider node $i$ and its two neighboring nodes $i-1$ and $i+1$. If either node $i-1$ or $i+1$ transmits a packet, while the other remains silent during the whole packet duration, the packet is assumed to be received with no error at node~$i$.
However, if both node $i-1$ and $i+1$ transmit and the two packets overlap either partially or completely, then both packets are considered lost and unrecoverable at node~$i$. In this case, we say that there is a collision at node~$i$.   As an example, suppose that node $i-1$ transmits a packet at time $t_0$. This packet is successfully received by node $i$ if and only if node $i$ is in receive mode between time $t_0$ and $t_0+T$, and node $i+1$ does not transmit any packet between time  $t_0 - T$ and $t_0+T$. We call this network a {\em tandem collision network}.

We adopt the approach in~\cite{MasseyMathys85} and impose the
restriction that the packet transmission times are independent of
the messages to be sent or forwarded, and independent of how the
other nodes access the channel. This can be accomplished by statically assigning each node a {\em protocol signal} for channel access. The protocol signal for node $i$, $s_i(t)$, is a deterministic and periodic signal of period $P_i$ second, and is equal to either zero or one for all~$t$. A protocol signal is equal to one  over some semi-open intervals whose lengths are integral multiple of~$T$. Node $i$ is required to transmit packets whenever $s_i(t)=1$, and remain silent whenever $s_i(t) = 0$. Due to the lack of common time reference among the nodes, the protocol signals are subject to delay offsets. We denote the delay offset of node $i$ by $\delta_i$, which takes value between 0 and $P_i$. We assume that the delay offsets are arbitrarily chosen but fixed throughout the communication session. Two protocol signals with delay offsets $\delta_1$ and $\delta_2$ are plotted in Fig.~\ref{fig:protocol_signal}. A packet is sent within the duration of each ``square pulse.'' We can see from Fig.~\ref{fig:protocol_signal} that the first two packets of both nodes are collided, but the third packet from node~2 can be received successfully at node~1.

We define the duty factor of $s_i(t)$ by
\begin{equation}
 f_i \triangleq \frac{1}{P_i} \int_0^{P_i} s_i(t) \,dt.
 \label{eq:def_duty_factor}
\end{equation}
It measures the fraction of time that node $i$ is transmitting.

\begin{figure}
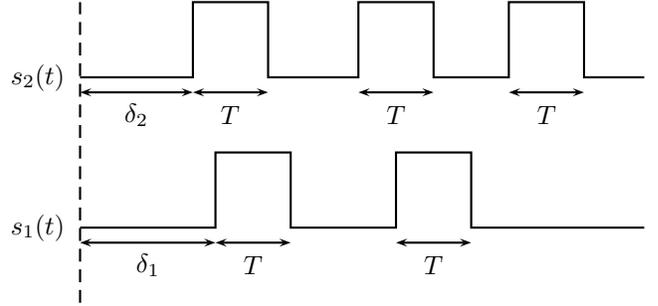

    \begin{center}
        \pspicture(-1,-1)(8,4)
    \psline[linestyle=dashed](0,-1)(0,3)

    \psline(0,0)(1.8,0)(1.8,1)(2.8,1)(2.8,0)(4.2,0)(4.2,1)(5.2,1)(5.2,0)(7.5,0)
    \psline(0,2)(1.5,2)(1.5,3)(2.5,3)(2.5,2)(3.7,2)(3.7,3)(4.7,3)(4.7,2)(5.7,2)(5.7,3)(6.7,3)(6.7,2)(7.5,2)

    \rput[r](-0.2,0){$s_1(t)$}
    \rput[r](-0.2,2){$s_2(t)$}
    \rput[u](0.9,-0.5){$\delta_1$}
    \psline[arrows=<->](0,-0.2)(1.8,-0.2)
    \rput[u](2.3,-0.5){$T$}
    \psline[arrows=<->](1.8,-0.2)(2.8,-0.2)
    \rput[u](4.7,-0.5){$T$}
    \psline[arrows=<->](4.2,-0.2)(5.2,-0.2)

    \rput[u](0.75,1.5){$\delta_2$}
    \psline[arrows=<->](0,1.8)(1.5,1.8)
    \rput[u](2,1.5){$T$}
    \psline[arrows=<->](1.5,1.8)(2.5,1.8)
    \rput[u](4.2,1.5){$T$}
    \psline[arrows=<->](3.7,1.8)(4.7,1.8)
    \rput[u](6.2,1.5){$T$}
    \psline[arrows=<->](5.7,1.8)(6.7,1.8)

            \endpspicture
    \end{center}
\caption{Protocol Signal and Relative Delay Offsets}
\label{fig:protocol_signal}
\end{figure}

We consider both slot-synchronous  and slot-asynchronous systems. In the {\em slot-synchronous} case, the delay offsets $\delta_i$, for $i=1,2,\ldots, M$, are integral multiples of packet duration~$T$. In the {\em slot-asynchronous case}, the delay offset $\delta_i$ is an arbitrary real number between 0 and $P_i$, for $i=1,2,\ldots, M$.
A time interval in the form $[kT, (k+1)T)$ for some integer $k$, with reference to the local clock, is called a {\em time slot}.
We say that a system is {\em time-slotted} if each packet is sent within a time slot. This means that the protocol signal in a time-slotted system satisfies
\begin{equation}
 s_i(t) = s_i( \lfloor t/T \rfloor T), \label{eq:time_slotted}
\end{equation}
where $\lfloor x \rfloor$ refers to the smallest integer larger than or equal to~$x$. It can be easily deduced from~\eqref{eq:time_slotted} that in a time-slotted system, the period $P_i$ of protocol signal $s_i(t)$ is an integral multiple of~$T$.

\begin{remark}
The notions of slot-synchronous system and time-slotted system are {\em not} the same and should be distinguished. ``Time-slotted'' is an attribute pertaining to the protocol signal set. ``Slot-synchronous'' is about synchronization of clocks among the nodes.
Fig.~\ref{fig:protocol_signal} provides an example that is neither  slot-synchronous nor time-slotted. It is not slot-synchronous because the difference between the two delay offsets, $\delta_2-\delta_1$, is not an integral multiple of~$T$. It is not time-slotted because the gap between the two pulses in $s_1(t)$ is not an integral multiple of~$T$.
\end{remark}

In a time-slotted but slot-asynchronous system, partial overlap of packets is inevitable. However, if a system is both time-slotted and slot-synchronous, then whenever two packets overlap, they overlap completely. For time-slotted system, we can compactly specify the protocol signals by binary sequences, called {\em protocol sequences}. A protocol signal of period $p_i T$ in a time-slotted system, where $p_i$ is an integer, corresponds to a discrete-time periodic sequence with period $p_i$, denoted by $s_i[k]$, such that
\[
 s_i[k] = 1 \text{ if and only if } s_i(t) = 1 \text{ for } t \in [kT, (k+1)T).
\]
We will abuse the use of language and denote both protocol signal and protocol sequence by the letter ``$s$''. Nevertheless, we can distinguish protocol signal and protocol sequences by using parenthesis for continuous time index and square bracket for discrete time index.

We assume that the protocol signals are jointly designed and known to all nodes, and will not be changed throughout the communication session in concern. In this paper, we do {\em not} assume any packet header, and hence a packet only carries data and does not contain any sender's information. As we will see in a later section, the sender of a packet can be identified by some sliding-window algorithm, which requires the knowledge neighboring nodes' protocol signals. We will also show that the maximal system throughput can be achieved without any packet header. Nevertheless, if packet headers are present, as in most practical systems, we can relax the requirement that the protocol signals are known to all nodes.

\smallskip

The block diagram of a node which is both a source and a destination is shown in Fig.~\ref{fig:block_diagram1}. The encoder produces packets as a function of the data source and all previously received packets. The transmission times of the packets are determined by the protocol signal generator, which  is a stand-alone device without any input from the source, the decoder, or the channel. The computation of  decoder's output is based on the content of the successfully received packets as well as the data from the local source. In Fig.~\ref{fig:block_diagram1}, flow of data is indicated by solid arrow, and flow of control signal is indicated by dashed arrow.
The block diagram of a source node is the same as in Fig.~\ref{fig:block_diagram1} except that the data sink is absent. Fig.~\ref{fig:block_diagram2} shows  the block diagram of a destination node. The output of the decoder is feedback and re-encoded. For a pure relay node, the block diagram is the same as in Fig.~\ref{fig:block_diagram2} except that there is no sink.

\begin{figure}
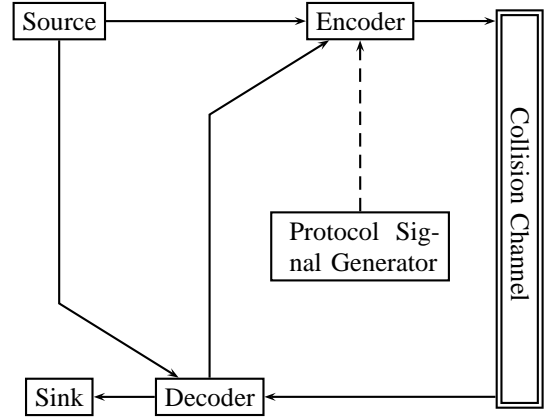

    \begin{center}
        \pspicture(-1,-1)(7,6)
\rput(0,0){\rnode{Sink}{\psframebox{Sink}}}
\rput(0,5){\rnode{Source}{\psframebox{Source}}}
\rput(2,0){\rnode{Dec}{\psframebox{Decoder}}}
\rput(4,5){\rnode{Enc}{\psframebox{Encoder}}}
\rput(4,2){\rnode{PSG}{\psframebox{ \parbox[c]{2cm}{Protocol Signal Generator}}}}

\ncline{->}{Source}{Enc}
\ncline{->}{Dec}{Sink}
\ncline[linestyle=dashed]{->}{PSG}{Enc}
\ncdiagg[angleA=90,armA=3.5cm, arrows=->]{Dec}{Enc}
\ncdiagg[angleA=-90,armA=3.5cm, arrows=->]{Source}{Dec}
\rput(5.8,0){\rnode{A}{}}
\ncline{->}{A}{Dec}
\rput(5.8,5){\rnode{B}{}}
\ncline{->}{Enc}{B}

\rput[l](5.8,2.5){\rnode{Channel}{\psdblframebox{\rotateright{\ \ \ \ \ \ \ \ \   Collision Channel \ \ \ \ \ \ \ \ \  }}}}

    \endpspicture
    \end{center}
\caption{Block Diagram of a Source and Destination Node}
\label{fig:block_diagram1}
\end{figure}

\begin{figure}
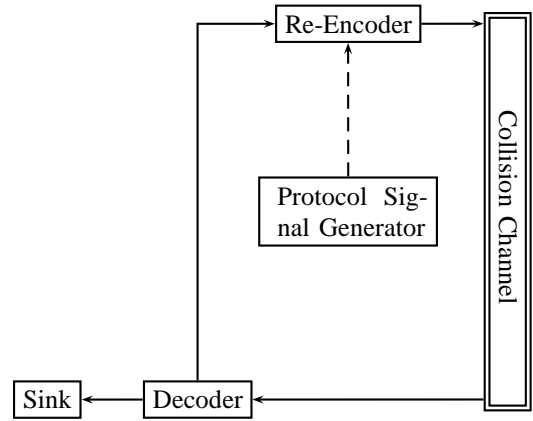

    \begin{center}
        \pspicture(-1,-1)(7,6)
\rput(0,0){\rnode{Sink}{\psframebox{Sink}}}
\rput(2,0){\rnode{Dec}{\psframebox{Decoder}}}
\rput(4,5){\rnode{Enc}{\psframebox{Re-Encoder}}}
\rput(4,2.5){\rnode{PSG}{\psframebox{ \parbox[c]{2cm}{Protocol Signal Generator}}}}

\ncline{->}{Dec}{Sink}
\ncline[linestyle=dashed]{->}{PSG}{Enc}
\ncangle[angleA=90,angleB=180,armA=3.5cm, arrows=->]{Dec}{Enc}
\rput(5.8,0){\rnode{A}{}}
\ncline{->}{A}{Dec}
\rput(5.8,5){\rnode{B}{}}
\ncline{->}{Enc}{B}

\rput[l](5.8,2.5){\rnode{Channel}{\psdblframebox{\rotateright{\ \ \ \ \ \ \ \ \   Collision Channel \ \ \ \ \ \ \ \ \ }}}}

    \endpspicture
    \end{center}
\caption{Block Diagram of a Destination Node}
\label{fig:block_diagram2}
\end{figure}

We describe the channel model formally as follows. We consider $nT$ seconds of transmission time, where $n$ is an integer and $T$ is the time duration of a packet. Each data source is discrete memoryless  over a $Q$-ary alphabet~$\Omega$.
For $j=1,2,\ldots, N$, source $j$ produces source symbols at a rate of $R_j$ symbols per packet duration. Each source symbol is chosen uniformly and independently from~$\Omega$.  Let $W_j$ denote a message from source~$j$. In a duration of $n$ packet durations.
$W_j$ may assume any value in $\mathcal{W}_j \triangleq \{1,2,\ldots, Q^{R_jn} \}$.
Since it is assumed that each packet can store a $Q$-ary symbol as well, we can store one source symbol in a packet.

For given delay offsets $\delta_1, \ldots, \delta_M$, the encoding and decoding functions are specified as follows.
If node $i$ is a source node, associated with source $j$, the packet transmitted by node $i$ at time $t$ is obtained by applying the encoding function
\[
 f_{i,t}: \mathcal{W}_j \times \Omega^{r_i(t)} \rightarrow \Omega,
\]
where $r_i(t)$ is the number of successfully received packets by node $i$ up to time~$t$. As the delay offsets are fixed, $r_i(t)$ is known and well-defined.
If node $i$ is not a source node, then the encoding function for the packet at time $t$ is
\[
 f_{i,t}: \Omega^{r_{i}(t)} \rightarrow \Omega.
\]

Suppose source $j$ is demanded by node $i$. If node~$i$  is associated  with another source, say source $j'$,  then at the end of $n$ packet durations, node~$i$ decodes $W_j$ by
\[
 g_{i,j} : \Omega_{j'} \times \Omega^{r_i(nT)} \rightarrow \Omega_j.
\]
If node $i$ is not a source node and wants to decode the data from source $j$, then the decoding function has the form
\[
 g_{i,j} : \Omega^{r_i(nT)} \rightarrow \Omega_j.
\]
For $i \in \beta(j)$, let the estimated value of the message from source $j$ by node $i$ be denoted by $\hat{W}_{ij}$. We say that there is a decoding error if  $W_j \neq \hat{W}_{ij}$ for some  $i \in \beta(j)$ and $j \in \{1,2,\ldots, N\}$.

A rate vector $(R_1,R_2, \ldots, R_N)$ with $R_j\geq 0$ for all $j$ is said to be {\em achievable} if there are $M$ protocol signals
$\{s_i(t)\}$, and encoding and decoding functions, such that for all $j$, source $j$ can be multicast to all nodes in $\beta(j)$ at a rate of $R_j$ symbols per packet duration with no decoding error, regardless of the delay offsets
$\delta_1, \ldots, \delta_M$. The closure of the set of all
achievable rate vectors, denoted by $\mathcal{C}_a(N,M,\alpha,\beta)$, is called the {\em zero-error capacity
region}, or simply the {\em capacity region}. The subscript ``$_a$'' signifies that the capacity region is for slot-asynchronous system. For slot-synchronous system,
we define achievable rate vector similarly, except that the delay offsets are restricted to integral multiples of $T$. The corresponding capacity region is denoted by $\mathcal{C}_s(N,M, \alpha, \beta)$.

We say that a tandem collision network is {\em bi-directional} if for all $j$ such that $1\neq \alpha(j)\neq M$, we can find $i$ and $i'$ in $\beta(j)$ such that $i<\alpha(j)< i'$. In words, it means that unless a source is associated with the left-most or right-most node, its message has to be sent to a node to the left and a node to the right. The two tandem collision networks in Examples 1 and 2 are both bi-directional. One of the main result in this paper is that, for bi-directional tandem collision network, the capacity region for both slot-synchronous and slot-asynchronous case can be achieved by the transmission scheme given in the next section.

\begin{remark}
Once the delay offsets $\delta_i$ are fixed, the transmission times of all packets in the future are also fixed.  Each link between two nodes becomes a deterministic erasure channel; the number of successfully received packets at each node per period is determined. So, for each fixed combination of delay offsets, the tandem collision network reduces to an {\em Aref network}~\cite{Aref80, Ratnakar06}. As we are interested in the worst-case throughput over all possible delay offsets, the analysis in this paper can be interpreted as taking the minimum throughput of a collection of Aref networks.
\end{remark}

\begin{remark} The assumption that
the protocol sequences are periodic is not restrictive, since in
practice, all pseduo-random number generators are periodic, and in
theory, the period can be arbitrarily large.
\end{remark}

\begin{remark}
We have the assumption that the transmission times of packets are independent of the messages to be transmitted or forwarded, and independent of the transmission of the other nodes.  The capacity region in this paper is obtained under the condition that channel accessing is done by protocol signal, so that collision avoidance algorithm and re-scheduling of packet transmission times etc. are not allowed.
\end{remark}

\section{A Transmission Scheme} \label{sec:achievable}

We first give a transmission scheme for slot-synchronous system, and then describe an extension to the slot-asynchronous case at the end of this section. The transmission scheme to be devised is time-slotted, and is based on a special class of protocol sequence, called shift-invariant protocol sequences, and a joint channel-network coding scheme. We remark that this transmission scheme can be applied to tandem collision network which may or may not be bi-directional. However, for bi-directional network, we will show in the next section that the achievable rates are indeed optimal.

As mentioned in Section~\ref{sec:system}, a protocol signal in time-slotted system can be specified  by a zero-one sequence, $s[k]$, so that a packet is transmitted in the $k$-th time slot $[kT, (k+1)T)$ if and only if $s[k] = 1$.
Let the smallest common period of the $M$ protocol sequences be denoted by integer~$P$. The duty factor of a zero-one sequence $s[k]$ of period $P$ is defined as
\[
 f_i \triangleq \frac{1}{P} \sum_{k=1}^P s[k].
\]
It can be easily seen that this is compatible with the notion of duty factor for continuous-time protocol signal.

\subsection{Slot-Synchronous System}

For slot-synchronous system, the delay offsets are all integral multiples of packet duration~$T$. With a slight abuse of language we model the delay offsets by integers, instead of real numbers. We will denote the delay offsets of node $i$ by an integer $\tau_i$, for $i=1,2,\ldots, M$. The actual delay is $\tau_i T$.

\paragraph{Shift-invariant Protocol Sequences} Shift-invariant protocol sequences are first used by Massey and Mathys~\cite{MasseyMathys85} to achieve the capacity region of the multiple-access collision channel without feedback. The following definition is from~\cite{SCSW09}. Given a set of $M$  zero-one sequences, $s_i[k]$, for $i=1,2,\ldots, M$, with common period $P$, and a subset
$\mathcal{A} = \{i_1, i_2,\ldots, i_m\} \subseteq \{1,2,\ldots, M\}$, the {\em generalized Hamming cross-correlation} is defined by
\[
H(\tau_1,\ldots, \tau_m; \mathcal{A}) \triangleq
\sum_{k=1}^P \prod_{\mu=1}^m s_{i_\mu}[k-\tau_\mu].
\]
We note that for a subset $\mathcal{A}$ of $\{1,2,\ldots, M\}$ consisting of two elements, the notion of generalized Hamming cross-correlation reduces to the usual Hamming cross-correlation for a pair of sequences. For $\mathcal{A}$ that is a singleton, it reduces to the Hamming weight. A set of $M$ protocol sequences
is said to be {\em shift-invariant} if for each subset $\mathcal{A} \subseteq \{1,2,\ldots, M\}$, the generalized Hamming cross-correlation is independent of delay offsets.

For the application in this paper, we do not need the full force of shift invariance. In a network with linear topology, any transmission from a node two or more hops away does not cause any interference. This justifies the restriction of our attention to generalized Hamming cross-correlation for subset $\mathcal{A}$ consisting of three or less elements. We say that a set of $M$ protocol sequences is {\em consecutively 3-wise shift-invariant} if
for each subset $\mathcal{A}$ of size three or less, consisting consecutive integers from $\{1,2,\ldots, M\}$, the generalized Hamming cross-correlation $H(\tau_1,\ldots, \tau_m; \mathcal{A})$ is independent of the delay offsets.

We define the {\em throughput} from node $i$ to node $i+1$, denoted by $\theta_{i,i+1}(\tau_i, \tau_{i+1}, \tau_{i+2})$, as the number of packets from node $i$ to node $i+1$ without collision in a period divided by~$P$,
 \begin{equation}
 \frac{1}{P} \sum_{k=1}^P s_i[k-\tau_i] (1-s_{i+1}[k-\tau_{i+1}]) (1-s_{i+2}[k-\tau_{i+2}]). \label{eq:throughput2}
\end{equation}
Here, $\tau_i$, $\tau_{i+1}$ and $\tau_{i+2}$ are integers representing the delay offsets of node $i$, $i+1$ and $i+2$, respectively.
We note that the throughput from node $i$ to node $i+1$ is affected only by the protocol signal of nodes $i$, $i+1$ and $i+2$, and hence is a function of $\tau_i$, $\tau_{i+1}$ and $\tau_{i+2}$.

The throughput from node $i$ to node $i-1$ is similarly defined as
\begin{equation}
 \frac{1}{P} \sum_{k=1}^P s_i[k-\tau_i] (1-s_{i-1}[k-\tau_{i-1}]) (1-s_{i-2}[k-\tau_{i-2}]) \label{eq:throughput3}
\end{equation}
and denoted by  $\theta_{i,i-1}(\tau_i, \tau_{i-1}, \tau_{i-2})$.

\begin{lemma}
If consecutively 3-wise shift-invariant protocol sequences are used in a time-slotted tandem collision network, then for all delay offsets $\tau_{i-2}$, $\tau_{i-1}$, $\tau_i$, $\tau_{i+1}$ and  $\tau_{i+2}$, we have
\[
\theta_{i,i+1}(\tau_i, \tau_{i+1}, \tau_{i+2}) =  f_i (1- f_{i+1}) (1 - f_{i+2})
\]
 and
\[
\theta_{i,i-1}(\tau_i, \tau_{i-1}, \tau_{i-2}) =  f_i (1- f_{i-1}) (1 - f_{i-2}),
\]
where $f_i$ is the duty factor of the $i$-th protocol sequence.
\label{lemma:SI}
\end{lemma}

The above lemma says that the throughput function does not depend on the relative delay offsets, provided that the protocol sequences are consecutively 3-wise shift-invariant.
The proof is based on an elementary property of zero-one sequences~\cite{SP80}, which is included here for the sake of completeness.

\begin{proof}
We prove only the statement for $\theta_{i,i+1}(\tau_i, \tau_{i+1}, \tau_{i+2})$. The proof for the second one is similar and omitted.

The summand in~\eqref{eq:throughput2} can be expanded as a linear combination of four terms, namely $s_i[k-\tau_i]$, $s_i[k-\tau_i]s_{i+1}[k-\tau_{i+1}]$,
$s_i[k-\tau_i]s_{i+2}[k-\tau_{i+2}]$, and
 \[
  s_i[k-\tau_i]s_{i+1}[k-\tau_{i+1}] s_{i+2}[k-\tau_{i+2}].
 \]
After taking the summation over $k$, we can see that $\theta_{i,i+1}(\tau_i,\tau_{i+1},\tau_{i+2})$ is equal to the sum of four generalized Hamming cross-correlations, each of which is independent of the delay offsets $\tau_i$'s by the shift-invariant assumption. Hence the linear combination is also independent of the delay offsets. This proves that $\theta_{i,i+1}(\tau_i,\tau_{i+1},\tau_{i+2})$ is independent of delay offsets. Let this value be denoted by~$\Theta$.

Next, we sum~\eqref{eq:throughput2} over $\tau_i$, $\tau_{i+1}$ and $\tau_{i+2}$. After exchanging  the order of summations, we obtain
\begin{align*}
P^3 \Theta
& = \frac{1}{P} \sum_{k=1}^P \sum_{\tau_i=1}^P s_i[k-\tau_i] \sum_{\tau_{i+1}=1}^P (1-s_{i+1}[k-\tau_{i+1}]) \\
& \qquad \cdot \sum_{\tau_{i+2}=1}^P (1-s_{i+2}[k-\tau_{i+2}]) \\
&=\frac{1}{P} \sum_{k=1}^P (P f_i) (P (1-f_{i+1})) (P (1-f_{i+2})) \\
& = P^3 f_i (1-f_{i+1}) (1-f_{i+2}).
\end{align*}
This proves the first part of the lemma.
\end{proof}

We write the duty factor of the protocol sequence $s_i[k]$ as $n_i/d$, where $d$ is a common denominator. We next  construct consecutively 3-wise shift-invariant sequences with period $d^3$. Given positive integer $n$ and $d$, let $\vec{u}(n,d)$ denote a $d$-dimensional row vector whose $n$ components on the left are 1, and the $(d-n)$ entries on the right are~0,
\[
 \vec{u}(n,d) \triangleq [ \underbrace{1\ 1\ \ldots 1}_{n}\ \underbrace{0\ 0\ \ldots 0}_{d-n} ].
\]

{\bf Construction:} Given $M$ fractions $n_i/d$ for $i=1,2,\ldots, M$, we construct $M$ protocol sequences of period $d^3$ as follows. For $i \equiv 1 \bmod 3$, let $s_i[k]$ be the concatenation of $d^2$ copies of $\vec{u}(n_i,d)$,
\[
 [\underbrace{\vec{u}(n_i,d)\ \vec{u}(n_i,d)\ \ldots \vec{u}(n_i,d)}_{d^2}].
\]
For $i \equiv 2 \bmod 3$, let $s_i[k]$ be the concatenation of $d$ copies of $\vec{u}(d n_i ,d^2)$,
\[
 [\underbrace{ \vec{u}(d n_i,d^2)\ \vec{u}(d n_i,d^2)\ \ldots \vec{u}(d n_i,d^2)}_{d}].
\]
For $i \equiv 0 \bmod 3$, let $s_i[k]$ be $\vec{u}(d^2 n_i ,d^3)$.

\begin{example}
Suppose that $M=5$ and the the duty factors are $f_i = 1/3$ for $i=1,2,3$, and $f_i = 2/3$ for $i= 4,5$. The protocol sequences constructed by the above method are
\begin{align*}
s_1[k]:\ & 100\ 100\ 100\ 100\ 100\ 100\ 100\ 100\ 100 \\
s_2[k]:\ & 111\ 000\ 000\ 111\ 000\ 000\ 111\ 000\ 000 \\
s_3[k]:\ & 111\ 111\ 111\ 000\ 000\ 000\ 000\ 000\ 000 \\
s_4[k]:\ & 110\ 110\ 110\ 110\ 110\ 110\ 110\ 110\ 110 \\
s_5[k]:\ & 111\ 111\ 000\ 111\ 111\ 000\ 111\ 111\ 000
\end{align*}
\end{example}

\begin{prop}
The protocol sequences by the above construction method are consecutively 3-wise shift-invariant. \label{prop:SI}
\end{prop}

The proof of Prop.~\ref{prop:SI} is similar to Theorem~8 in~\cite{SCSW09}, and is omitted.

\medskip

\paragraph{Identifying the senders of a packet}
We will use an interesting property of consecutively 3-wise shift-invariant sequences, called {\em identifiability}. Consider node~1. All successfully received packets at node~1 come from node~2. Also, it is also clear that all successfully received packets at node~$M$ come from node~$M-1$. Identifying the sender of packets is trivial for the two nodes at the ends.

For the nodes in between, it is not immediate to tell whether the packets are from the left or from the right. One method of providing the identity of senders is to attach a header at the beginning of each packet. However, if consecutively  3-wise shift-invariant protocol sequences are used, the receiver can identify the sender of all non-collided packets by observing the channel activities, without looking into the content of the packets. Hence theoretically, we can dispense with packet header.

After a period of $P$ time slots, a node records the status of channel in each time slot. Consider a given node, say node $i$. We indicate the time slot when node $i$ is transmitting by $\Delta$. For the time slot when node $i$ is receiving, we use symbols ``0'', ``1'' and ``$*$'' to indicate whether there is 0, 1, or more than 1 packets is received. The symbol ``$*$'' represents a collision.

We call this sequence of ``$\Delta$'',``0'', ``1'' and ``$*$'' the {\em channel activity signal}. If we can always determine the senders of all non-collided packets from the channel activity signal, regardless of what the delay offsets are, then we say that the protocol sequence set is {\em identifiable}. As an example, suppose that all delay offsets are zero, and the protocol sequences in Example~3 are used. The channel activity signal observed at node 2 is
\begin{equation}
\Delta\!\Delta\!\Delta *\!11 *\!11\ \Delta\!\Delta\!\Delta\ 100\ 100\ \Delta\!\Delta\!\Delta\ 100\ 100.
\label{eq:identifiability}
\end{equation}
If the protocol sequences are identifiable, the receiver at node~2 is able to deduce from the above channel activity signal that the packets at time slots 5, 6, 8 and 9 are from node~3, and the packets at time slots 13, 16, 22 and 25 are from node~1. We remark that identifiability does not mean that the delay offsets are determined. In fact, for the channel activity signal in~\eqref{eq:identifiability}, we can cyclically shift $s_1[k]$ by any multiple of~3, without any change in the channel activity signal. Therefore, determining the delay offsets uniquely using only the channel activity signal is not possible in general.

We now show that the property of identifiability is implied by consecutively 3-wise shift invariance. Let $c_i[k]$ denote the channel activity signal observed by node~$i$. Suppose that the delay offsets of nodes $i-1$ and $i+1$ are $\tau_{i-1}$ and $\tau_{i+1}$, which is unknown to node~$i$. Given $\tau_{i-1}'$ and $\tau_{i+1}'$ from 0 to $P-1$, node $i$ compute
\[
 c_i'[k] = \begin{cases}
 \Delta & \text{if } s_i[k+\tau_i] = 1 \\
 0 & \text{if } s_i[k+\tau_i] = 0 \text{ and}\\
 &  \qquad s_{i-1}[k+\tau_{i-1}'] =s_{i+1}[k+\tau_{i+1}'] =0, \\
 * & \text{if } s_i[k+\tau_i] = 0 \text{ and} \\
 & \qquad s_{i-1}[k+\tau_{i-1}'] =s_{i+1}[k+\tau_{i+1}'] =1, \\
 1 & \text{if } s_i[k+\tau_i] = 0 \text{ and either} \\
  & \ \ \text{$s_{i-1}[k+\tau_{i-1}']=1$  or $s_{i+1}[k+\tau_{i+1}']=1$}. \\
 \end{cases}
\]
Node $i$ would observe $c_i'[k]$ if the delay offsets of node $i-1$ and $i+1$ were $\tau_{i-1}'$ and $\tau_{i+1}'$ respectively. We search for $\tau_{i-1}'$ and $\tau_{i+1}'$ such that the corresponding $c_i'[k]$ is the same with the true channel activity signal $c_i[k]$. There is always at least one solution to this search problem, because $\tau_{i-1}'= \tau_{i-1}$ and $\tau_{i+1}' = \tau_{i+1}$ is one such solution. In general, there may be multiple solutions. After one such pair of $(\tau_{i-1}', \tau_{i+1}')$ is found, we
then declare that the packet with time index $k$ satisfying $c_i'[k] = s_{i\pm 1 }[k+\tau_{i\pm 1}']=1$ is sent from node $i\pm 1$.

The correctness of this algorithm, provided that the protocol sequences $s_{i-1}[k]$, $s_{i}[k]$ and $s_{i+1}[k]$ are consecutively 3-wise shift-invariant, is given in the appendix.

\medskip

\paragraph{Determining the delay offsets}
We have seen that the identifiability property is insufficient to determine the delay offsets. Nevertheless, it is pointed out in~\cite{MasseyMathys85} that we can find the delay offsets by the following initialization mechanism at the beginning of the communication session.

Suppose a node began to transmit information starting from some finite time in the past, and before that, the zero packet was transmitted in the infinite past when a node is required to transmit a packet by the protocol sequence. Following the notation in~\cite{MasseyMathys85}, we call the $Pf_i$ packets from node~$i$ within a period of $P$ time slots as a {\em frame}.
Before any information packets are transmitted, node~$i$ first sends the following $Pf_i+1$ frames $[1,1,\ldots, 1]$, $[1,0,\ldots, 0]$, $[0,1,0\ldots, 0] , \ldots, [0,\ldots, 0, 1]$. In the first frame, all packets contain the value $1$ in the alphabet set $\Omega$. In each of the remaining $P f_i$ frames, there is exactly one packet which equals~1.

We now describe how node $i+1$ determine the delay offset of node~$i$. Before node $i$ starts transmitting, node $i+1$ can only observe idle time slots, or packets from node $i+2$. When node $i+1$ first receives a packet containing the symbol~1 from user~$i$, node $i+1$ buffers $P f_i+1$ periods of packets. (Node $i+1$ can distinguish the packets from node $i$ and $i+2$ by the identifiability property of shift-invariant sequences.) Node $i+1$ then tries to find a pair of packets from node~$i$, containing the value ``1'', such that their time difference is an integral multiple of~$P$. Suppose that there are two such packets separated by $a P$ time indices for some integer $a$. The first packet should belong to the very first frame $[1,\ldots, 1]$, and the second to the frame $[0,\ldots, 0, 1, 0, \ldots, 0]$ with a ``1'' in the $a$-th position. Since node $i+1$ knows the protocol sequence of node~$i$, $s_i[k]$, the time index of the $a$-th ``1'' in $s_i[k]$ is also known.  The time indices of this pair of packets, reduced modulo $P$, is the relative delay offset of node $i$.

\medskip

\paragraph{Joint Channel-Network Coding}
The encoder in each node has two objectives: combat the erasures caused by collisions (channel coding) and facilitate information flow in both directions (network coding). This is achieved by a joint channel-network coding scheme called {\em nested coding}, which is a coding technique found useful in cooperative relaying~\cite{Yang05, Xiao06, Wu07}. The main difference between the nested coding scheme in this paper and those in~\cite{Yang05, Xiao06, Wu07 } is that our target is to correct erasures, while the nested codes in~\cite{Yang05, Xiao06, Wu07} are aimed at correcting errors.

We continue the notation that the duty factor at node $i$, $f_i$, is equal to a fraction $n_i/d$, for $i=1,2,\ldots, M$.
With the use of the consecutively 3-wise shift-invariant protocol sequences constructed earlier, the common period of the protocol sequences is $P=d^3$, and
the number of packets sent out by node $i$ in a period is $n_id^2$. In this section, we assume that $Q$ is a prime power and $Q \geq n_i d^2$ for $i=1,2,\ldots, M$, and assume that $\Omega$ is the finite field of $Q$ elements. Let the element of $\Omega$ be ordered in some arbitrary way,
\[
 \Omega = \{ \omega_1, \omega_2, \ldots, \omega_Q\}.
\]

We use Reed-Solomon (RS) code~\cite{RS60} as a building block of the nested coding. Let $\mathcal{F}_k$ be the set of all polynomials over $\Omega$ with degree less than or equal to $k-1$,
\[
 \mathcal{F}_k \triangleq \Big\{ \sum_{\ell = 0}^{k-1} c_\ell x^\ell:\, c_\ell \in \Omega \Big\}.
\]
The coefficients of a polynomial in $\mathcal{F}_k$ are treated as information symbols to be encoded.
Given any subset $\mathcal{X}$ of $\Omega$ with $|\mathcal{X}|$ elements, we let
$(f(\omega))_{\omega \in \mathcal{X}}$ be the $|\mathcal{X}|$-tuple with the component indexed by $\omega$ equal to the value of $f$ evaluated at~$\omega$.  The vector $(f(\omega))_{\omega \in \mathcal{X}}$ is called a {\em codeword}.
We define the $k$-dimensional RS code on $\mathcal{X}$ over $\Omega$ as the set of all codewords
\[
 \{ (f(\omega))_{\omega \in \mathcal{X}} \in \Omega^{|\mathcal{X}|} :\, f \in \mathcal{F}_k \}.
\]
The encoding function maps a polynomial of degree strictly less than $k$, with coefficient in $\Omega$, to $(f(\omega))_{\omega \in \mathcal{X}}$. The encoding map is denoted by $\ev_\mathcal{X}(f)$.  For any $f \in \mathcal{F}_k$, if we know the value of $f$ evaluated at any $k$ elements in $\mathcal{X}$, then we can recover $f$ by interpolation, and thus the coefficients of $f$ can be uniquely determined. In fact, given $k$ points
\[
 (\omega_{t_1}, f(\omega_{t_1})),  (\omega_{t_2}, f(\omega_{t_2})), \ldots,  (\omega_{t_k}, f(\omega_{t_k})),
\]
where $\omega_{t_i}$ are distinct for $i=1,2,\ldots, k$, we can obtain the coefficients of $f(x) = \sum_{\ell=0}^{k-1} c_\ell x^{\ell}$ by solving
\[
\begin{bmatrix}
1 & \omega_{t_1} & \omega_{t_1}^2 & \cdots &\omega_{t_1}^{k-1} \\
1 & \omega_{t_2} & \omega_{t_2}^2 & \cdots &\omega_{t_2}^{k-1} \\
\vdots & \vdots & \vdots & \ddots & \vdots \\
1 & \omega_{t_k} & \omega_{t_k}^2 & \cdots &\omega_{t_k}^{k-1}
\end{bmatrix}
\begin{bmatrix}
c_0 \\ c_1 \\ \vdots \\ c_{k-1}
\end{bmatrix}
=\begin{bmatrix}
f(\omega_{t_1}) \\ f(\omega_{t_2}) \\ \vdots \\ f(\omega_{t_k})
\end{bmatrix}.
\]
The matrix on the left hand side is a Vandermonde matrix, and is invertible because $\omega_{t_1}$, $\omega_{t_2}, \ldots, \omega_{t_k}$ are distinct.

To facilitate the discussion on the encoding and decoding procedure, we make the following definitions.
For $i=2,3,\ldots, M-1$, let $\mathcal{S}_{i}^f$ be the set of sources associated with a node to the left of node $i$ and demanded by a node to the right of node~$i$,
\begin{equation}
 \mathcal{S}_{i}^f \triangleq \{ j:\, \alpha(j) < i, \text{ and }  \beta(j) \ni i' > i\}.  \label{eq:source_left_right}
\end{equation}
Also, for $i=2,3,\ldots, M-1$, let $\mathcal{S}_{i}^b$ be the set of sources associated with a node to the right of node $i$ and demanded by a node to the left of node~$i$,
\begin{equation}
 \mathcal{S}_{i}^b \triangleq \{ j:\, \alpha(j) > i, \text{ and }  \beta(j) \ni i' < i\}. \label{eq:source_right_left}
\end{equation}
We set $\mathcal{S}_{1}^f = \mathcal{S}_{M}^f = \mathcal{S}_{1}^b = \mathcal{S}_{M}^b = \emptyset$.

The superscripts ``$^f$'' and ``$^b$'' stand for ``forward'' and ``backward'' respectively. We note that $\mathcal{S}_{i}^f$ and $\mathcal{S}_{i}^b$ may be empty. For $i=2,3,\ldots, M-1$, let
\[r_{i}^f \triangleq \sum_{j \in \mathcal{S}_{i}^f} R_j\]
be the data rate through node $i$ from left to right, and
\[r_{i}^b \triangleq \sum_{j \in \mathcal{S}_{i}^b} R_j\]
be the data rate through node~$i$ from right to left.

To illustrate the notation, we consider the four-node network in Example~1.  The data from source~1 passes through nodes 2 and 3 in the forward direction, and the data from source~2 passes through nodes 2 and 3 in the backward direction. We have $\mathcal{S}_1^f = \mathcal{S}_1^b = \mathcal{S}_4^f =\mathcal{S}_4^b = \emptyset$, $\mathcal{S}_2^f = \mathcal{S}_3^f = \{1\}$, $\mathcal{S}_2^b = \mathcal{S}_2^b = \{2\}$, $r_2^f = r_3^f = R_1$ and $r_2^b = r_3^b = R_2$.

The encoding scheme is a decode-and-forward scheme. To simplify the description, we assume that all delay offsets $\tau_i$ are equal to zero, i.e., all protocol sequences are aligned, and consider the operations at node~$i$. In a period of $P$ packet durations, $r_i^f P$ source symbols which are going to be sent through node $i$ in the forward direction are generated. Call these symbols
\[
 a(1), a(2), \ldots, a(r_{i}^f P).
\]
Likewise, let
\[
 b(1), b(2), \ldots, b(r_{i}^b P),
\]
be source symbols to be sent to node~$i-1$ through node $i$ in the backward direction. Suppose that $\alpha(\sigma)=i$, i.e., source $\sigma$ is associated with node~$i$. In a time period of $PT$ seconds, source $\sigma$ produces $R_\sigma P$ symbols
\[
 c(1), c(2), \ldots, c(R_\sigma P).
\]
The encoder at node $i$ maps these $(r_{i}^f + r_{i}^b + R_\sigma)P$ symbols to $n_i d^2$ symbols and send them out according to a protocol sequence with duty factor $n_i/d$.
Define the following three polynomials,
\begin{align*}
 g(x)   & \triangleq  \sum_{k=1}^{R_\sigma P} c(k)x^{k-1} \\
 h^f(x) & \triangleq  \sum_{k=1}^{r_{i}^f P} a(k)x^{k-1} \\
 h^b(x) & \triangleq  \sum_{k=1}^{r_{i}^b P} b(k)x^{k-1} .
\end{align*}
The degrees of polynomials $g(x)$, $h^f(x)$ and $h^b(x)$ are no more than $R_\sigma P-1$, $r_{i}^f P-1$ and $r_{i}^b P-1$, respectively. Let $\mathcal{X}_i$ be the set of the first $n_id^2$ elements in $\Omega$. We transmit the following frame of packets
\[
  \ev_{\mathcal{X}_i}(g(x)+ (h^f(x) +h^b(x))x^{R_\sigma P} )
\]
in a period. Here, the addition and multiplication are polynomial arithmetics over the finite field~$\Omega$.

\begin{figure}
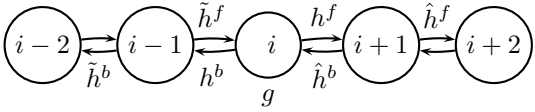

    \begin{center}
        \pspicture(-.3,-.5)(6,1)
     \cnodeput(0,0){N1}{$i-2$}
     \cnodeput(1.5,0){N2}{$i-1$}
     \cnodeput(3,0){N3}{$\phantom{+}i\phantom{1}$}
     \cnodeput(4.5,0){N4}{$i+1$}
     \cnodeput(6,0){N5}{$i+2$}

     \ncarc{->}{N1}{N2}
     \ncarc{->}{N2}{N1}
     \ncarc{->}{N2}{N3}
     \ncarc{->}{N3}{N2}
     \ncarc{->}{N3}{N4}
     \ncarc{->}{N4}{N3}
     \ncarc{->}{N4}{N5}
     \ncarc{->}{N5}{N4}

     \rput[u](3,-.7){$g$}
     \rput[u](2.25,0.4){$\tilde{h}^f$}
     \rput[u](3.75,0.4){$h^f$}
     \rput[u](5.25,0.4){$\hat{h}^f$}
     \rput[d](.75,-0.4){$\tilde{h}^b$}
     \rput[d](2.25,-0.4){$h^b$}
     \rput[d](3.75,-0.4){$\hat{h}^b$}
        \endpspicture
    \end{center}
\caption{Illustration for the Joint Channel-Network Coding}
\label{fig:encoding}
\end{figure}

Suppose source $\tilde{\sigma}$ is attached to node $i-1$ and produces symbols at a rate of $R_{\tilde{\sigma}}$ symbols per packet duration. At the end of a period, node~$i$ receives from node $i-1$ the following frame of packets
\begin{equation}
  \ev_{\mathcal{X}_{i-1}}( \tilde{g}(x)+ (\tilde{h}^f(x) + \tilde{h}^b(x))x^{R_{\tilde{\sigma}} P} )
  \label{eq:from_left}
\end{equation}
with some of the packets erased due to collision.

The degrees of $\tilde{g}(x)$, $\tilde{h}^f(x)$ and $\tilde{h}^b(x)$ are no more than $R_{\tilde{\sigma}}P-1$, $\tilde{r}_{i-1}^f P-1$ and $\tilde{r}_{i-1}^b P-1$, respectively. The coefficients of $\tilde{h}^f(x)$ are the symbols sent from node~$i-1$ to node~$i$, and the coefficients of $\tilde{h}^b(x)$ are the symbols sent from node~$i-1$ to node~$i-2$. The coefficients of $\tilde{g}(x)$ are the symbols from source $\tilde{\sigma}$. Since $\tilde{h}^b(x)$ is known to node~$i$, the decoder of node~$i$ can subtract
\[
 \ev_{\mathcal{X}_{i-1}}( \tilde{h}^b(x)x^{R_{\tilde{\sigma}} P} )
\]
from~\eqref{eq:from_left}, and obtain
\begin{equation}
 \ev_{\mathcal{X}_{i-1}}( \tilde{g}(x)+ \tilde{h}^f(x) x^{R_{\tilde{\sigma}} P} ),
 \label{eq:RS_codeword}
\end{equation}
with some components missing due to erasures. Here, the subtraction is done using arithmetics in finite field $\Omega$.
The vector in~\eqref{eq:RS_codeword} is an RS codeword, corresponding to a polynomial of degree no more than $(R_{\tilde{\sigma}}+\tilde{r}_{i-1}^f)P-1$. Provided that the number of non-collided packets is larger than or equal to $(R_{\tilde{\sigma}}+\tilde{r}_{i-1}^f)P$, then node $i$ can recover $\tilde{g}(x)$ and~$\tilde{h}^f(x)$.

Suppose there is no source associated with node $i+1$. Then the frame of packets transmitted by node $i+1$ is in the form of
\begin{equation}
  \ev_{\mathcal{X}_{i+1}}( \hat{h}^f(x) + \hat{h}^b(x) )
  \label{eq:from_right}
\end{equation}
The coefficients of $\hat{h}^f(x)$ are the symbols sent from node~$i+1$ to node~$i+2$, and the coefficients of $\hat{h}^b(x)$ are the symbols sent from node~$i+1$ to node~$i$. Since $\hat{h}^f(x)$ is known to node~$i$, the decoder of node~$i$ subtracts
$ \ev_{\mathcal{X}_{i+1}}( \hat{h}^f(x))$
from \eqref{eq:from_right}, and obtain an erased version of
 $\ev_{\mathcal{X}_{i+1}}( \hat{h}^b(x)  )$.
Note that $ \hat{h}^b (x)$ is a polynomial of degree no more than
$\hat{r}_{i+1}^b P - 1$.
Provided that node $i$ receives at least $ \hat{r}_{i+1}^b P$ non-collided packets from node~$i+1$, then the decoder of node $i$ can recover $\hat{h}^b(x)$. (See Fig.~\ref{fig:encoding}.)

After the coefficients of $\tilde{g}(x)$,  $\tilde{h}^f(x)$ and $\hat{h}^b(x)$ are recovered, the encoder of node~$i$ produces the frame of packets for the next period, and the process continues.

If the delay offsets are not zero, then the nodes need to buffer the decoded symbols for one period. With extra delay due to the buffering, the joint channel-network coding scheme works in a similar way as in the case with zero delay offset. We characterize the rate region achieved by this transmission scheme in the following theorem.

\begin{theorem}
For $i=1,2,\ldots, M$, let $\mathcal{S}_{i}^f$ and $\mathcal{S}_{i}^b$ be defined as in~\eqref{eq:source_left_right} and~\eqref{eq:source_right_left}.
A rate vector $(R_1,\ldots, R_N)$ is achievable in a slot-synchronous tandem collision network if for $i=1,2,\ldots, M$,
\begin{align}
R_\sigma + \sum_{j \in \mathcal{S}_{i}^f} R_j &\leq f_i(1-f_{i+1}) (1- f_{i+2}) \label{eq:condition_source1} \\
R_\sigma + \sum_{j \in \mathcal{S}_{i}^b} R_j &\leq f_i(1-f_{i-1}) (1- f_{i-2})
\label{eq:condition_source2}
\end{align}
when node $i$ is a source node associated with source $\sigma$, and
\begin{align}
 \sum_{j \in \mathcal{S}_{i}^f} R_j &\leq f_i(1-f_{i+1}) (1- f_{i+2}) \label{eq:condition_nosource1} \\
  \sum_{j \in \mathcal{S}_{i}^b} R_j &\leq f_i(1-f_{i-1}) (1- f_{i-2})
  \label{eq:condition_nosource2}
\end{align}
when node~$i$ is not a source node, for some non-negative real numbers $f_1, \ldots, f_M$  between 0 and~1.
\label{thm:synchronous}
\end{theorem}

(In Theorem~\ref{thm:synchronous}, $f_{-2}$, $f_{-1}$, $f_{M+1}$ and $f_{M+2}$ are set to zero.)

\begin{proof}
Let $(R_1, R_2, \ldots, R_N)$ be a rate vector that satisfies~\eqref{eq:condition_source1} to~\eqref{eq:condition_nosource2} for some real numbers $f_1$, $f_2,\ldots, f_M$ between 0 and~1.

For $i=1,2,\ldots, M$, we approximate $f_i$ by rational $\bar{f}_i = n_i/d$, where $d$ is a common denominator.  Construct $M$ consecutively 3-wise shift-invariant protocol sequences by the construction described at the beginning of this section, with common period $d^3$ and duty factor~$\bar{f}_i$, $i=1,2,\ldots, M$.

Consider the link from node $i$ to node $i+1$.
Node $i+1$ use the initialization mechanism described in paragraph (c) earlier in the section to determine the delay offset pertaining to node~$i$. This mechanism can always find the delay offset because the protocol sequences are consecutively 3-wise shift-invariant, and hence identifiable. After the delay offset of node $i$ is known, node $i+1$ can determine the time indices of the packets sent from node~$i$. By Lemma~\ref{lemma:SI},
node $i+1$ always receives $d^3 \bar{f}_i(1-\bar{f}_{i+1})(1-\bar{f}_{i+2})$ successful packets per slot duration from node~$i$. We recall that by the shift-invariant property, the number of successful packets from node $i$ to node $i+1$ in a period of $d^3$ slot times is a constant.

If node $i$ is a source node, with source $\sigma$ attached to it, the number of information packets from node $i$ to node $i+1$ in a period of $d^3$ time slots is $d^3(R_\sigma + \sum_{j\in \mathcal{S}_i^f} R_j)$. Provided that condition
\begin{equation}
R_\sigma + \sum_{j \in \mathcal{S}_{i}^f} R_j \leq \bar{f}_i(1-\bar{f}_{i+1}) (1- \bar{f}_{i+2}) \label{eq:condition_source1a}
\end{equation}
is satisfied, the channel-network coding scheme described earlier is able to recover the information packets with zero decoding error.
On other hand, if node $i$ is not a source node,
the joint channel-network coding scheme can decode the messages with zero error provided that
\begin{equation} \sum_{j \in \mathcal{S}_{i}^f} R_j \leq \bar{f}_i(1-\bar{f}_{i+1}) (1- \bar{f}_{i+2}). \label{eq:condition_nosource1a}
\end{equation}

The rate constraint for the link from node $i$ to node $i-1$ is either
\begin{align}
R_\sigma + \sum_{j \in \mathcal{S}_{i}^b} R_j &\leq \bar{f}_i(1-\bar{f}_{i-1}) (1- \bar{f}_{i-2})
\label{eq:condition_source2a}
\end{align}
or
\begin{align}
  \sum_{j \in \mathcal{S}_{i}^b} R_j &\leq \bar{f}_i(1-\bar{f}_{i-1}) (1- \bar{f}_{i-2}),
  \label{eq:condition_nosource2a}
\end{align}
depending on whether node $i$ is a source node or not.

Since $\bar{f}_i$ is an approximation to $f_i$, we can find a rate vector $(\bar{R}_1, \ldots, \bar{R}_N)$ close to $(R_1, \ldots, R_N)$, which satisfies~\eqref{eq:condition_source1a} to~\eqref{eq:condition_nosource2a}. The deviation of  $(\bar{R}_1,  \ldots, \bar{R}_N)$ from  $(R_1,  \ldots, R_N)$ depends on the differences between $\bar{f}_i$ and $f_i$, $i=1,2,\ldots, M$.
As we take $\bar{f}_i$ approaching $f_i$ for $i=1,2,\ldots, N$, the corresponding rate vector $(\bar{R}_1, \ldots, \bar{R}_N)$ approaches $(R_1, \ldots, R_N)$. Hence $(R_1, \ldots, R_N)$ is a limit point of a sequence of achievable rate vectors. This proves that $(R_1, \ldots, R_N)$ lies in the achievable rate region.
\end{proof}

\subsection{Extension to Slot-Asynchronous System} \label{sec:achievable_unslotted}

The coding scheme described above can be modified as in~\cite{MasseyMathys85} and operates in slot-asynchronous system with a slight loss of data rates.
The idea is to replace each zero in a protocol sequence by $m$ consecutive zeros, and each one by $m-1$ consecutive ones followed by a single zero. All duty factors are then multiplied by a factor $(m-1)/m$ after this process. The codewords from the joint channel-network coding scheme are interleaved $m-1$ times. For example, if $m=3$, the two protocol sequences $[1\ 0\ 1\ 0]$ and $[1\ 1\ 0\ 0]$ are mapped to two sequences of length~12,
\begin{align*}
&[1\ 1\ 0\ 0\ 0\ 0\ 1\ 1\ 0\ 0\ 0\ 0\ 0] \\
&[1\ 1\ 0\ 1\ 1\ 0\ 0\ 0\ 0\ 0\ 0\ 0\ 0].
\end{align*}

Let $\mathbf{R} = (R_1,\ldots, R_N)$ be a rate vector which is achievable when the system is slot-synchronous. It is shown in~\cite[Lemma 5]{MasseyMathys85} that the resulting transmission scheme is error-free in the slot-asynchronous case with rate $\frac{m-1}{m} \mathbf{R}$. We refer the readers to~\cite{MasseyMathys85} for the details of argument.
Since $m$ can be arbitrarily large, we conclude that the rate vector $\mathbf{R}$ is achievable in the slot-asynchronous case.
We have thus proved the following.

\begin{theorem}
Any rate vector that satisfies the conditions in~Theorem~\ref{thm:synchronous} is also achievable in the slot-asynchronous case.
\label{thm:asyn}
\end{theorem}

Theorem~\ref{thm:asyn} says that there is essentially no loss of achievable data rates when we compare slot-synchronous and slot-asynchronous system.
Nevertheless, if we approach the boundary of the achievable rate region by increasing the value of $m$, decoding delay also increases.

\section{Outer Bound on Capacity Region} \label{sec:converse}

In this section, we derive an outer bound on the achievable rate region. In fact, we will give an outer bound on achievable rate vectors for slot-synchronous system. This also yields an outer bound for slot-asynchronous system. We recall that a slot-synchronous system is not necessarily a time-slotted system. In the derivation of the outer bound, we do {\em not} assume that the system is  time-slotted.

For $i=1,2,\ldots, M-1$, let $\bar{\mathcal{S}}_{i}^f$ be the set of sources associated with node $i$ or a node to the left of node $i$, and demanded by a node to the right of node~$i$,
\begin{equation}
 \bar{\mathcal{S}}_{i}^f \triangleq \{ j:\, \alpha(j) \leq i, \text{ and }  \beta(j) \ni i' > i\}.  \label{eq:source_left_right2}
\end{equation}
For $i=2,3,\ldots, M$, let $\bar{\mathcal{S}}_{i}^b$ be the set of sources associated with node~$i$ or a node to the right of node~$i$, and demanded by a node to the left of node~$i$,
\begin{equation}
 \bar{\mathcal{S}}_{i}^b \triangleq \{ j:\, \alpha(j) \geq i, \text{ and }  \beta(j) \ni i' < i\}. \label{eq:source_right_left2}
\end{equation}
$\bar{\mathcal{S}}_M^f$ and $\bar{\mathcal{S}}_1^b$ are defined as the empty set.
In contrast to the definition of $\mathcal{S}_i^f$ and $\mathcal{S}_i^b$ in \eqref{eq:source_left_right} and~\eqref{eq:source_right_left}, we have ``$\leq$'' and ``$\geq$'' in~\eqref{eq:source_left_right2} and~\eqref{eq:source_right_left2} instead of strict inequality in~\eqref{eq:source_left_right} and~\eqref{eq:source_right_left};  both $\bar{\mathcal{S}}_{i}^f$ and $ \bar{\mathcal{S}}_{i}^b$ include the source associated with node~$i$. It is easy to see that $\bar{\mathcal{S}}_{i}^f \supseteq  \mathcal{S}_{i}^f$ and
$\bar{\mathcal{S}}_{i}^b \supseteq  \mathcal{S}_{i}^b$.

Given source mapping $\alpha$ and receiver mapping $\beta$ in a tandem collision network with $M$ nodes and $N$ sources,
let
\begin{align}
\mathcal{C}_{out}(M,N,\alpha,\beta) &\triangleq \bigcup \Big\{ (R_1,\ldots, R_N) \in \mathbb{R}_+^N: \notag \\
 \sum_{j \in \bar{\mathcal{S}}_a^f} R_j &\leq  f_{a}(1-f_{a+1})(1-f_{a+2}),  \label{eq:upper_bound1} \\
 \sum_{j \in \bar{\mathcal{S}}_a^b} R_j &\leq  f_{a}(1-f_{a-1})(1-f_{a-2}),  \label{eq:upper_bound2} \\
 &\text{ for } a=1,\ldots, M  \Big\} \notag
\end{align}
with the union taken over all real numbers $0 \leq f_i \leq 1$, $i=1,2,\ldots, M$. (The variables $f_{-2}$, $f_{-1}$, $f_{M+1}$, $f_{M+2}$ are taken to be zero.) The main result in this section is the following.

\begin{theorem}
All achievable rate vectors for slot-synchronous tandem collision network are contained in $\mathcal{C}_{out}(M,N,\alpha,\beta)$.
\label{thm:converse}
\end{theorem}

For a {\em bi-directional} tandem collision network, if node $i$ is a source node associated with source $\sigma$, then $\bar{\mathcal{S}}_i^f = \mathcal{S}_i^f \cup \{\sigma\}$ and
$\bar{\mathcal{S}}_i^b = \mathcal{S}_i^b \cup \{\sigma\}$ for $i=1,2,\ldots, M$. If node $i$ is not a source node, then $\bar{\mathcal{S}}_i^f = \mathcal{S}_i^f$ and $\bar{\mathcal{S}}_i^b = \mathcal{S}_i^b$.
Comparing the outer bound~\eqref{eq:upper_bound1} and~\eqref{eq:upper_bound2} and the rate constraints \eqref{eq:condition_source1} to \eqref{eq:condition_nosource2} in Theorem~\ref{thm:synchronous}, we see that the outer bound $\mathcal{C}_{out}(M,N,\alpha,\beta)$ coincides with the
achievable rate region for slot-synchronous and slot-asynchronous system described in the previous section. We have thus found the capacity region for bi-directional tandem collision network.

\begin{corollary}
For bi-directional tandem collision network, the capacity region in the slot-asynchronous and slot-synchronous case is equal to $\mathcal{C}_{out}(M,N,\alpha,\beta)$.
\label{thm:capacity}
\end{corollary}

\begin{proof}[Proof of Theorem~\ref{thm:converse}]
The proof roughly follows the same line as in~\cite[Section IV]{MasseyMathys85}.
We assume that for all $i$, the period $P_i$ of protocol signal $s_i(t)$ is a rational multiple of the duration of a packet, i.e., $P_i = (a_i/b_i) T$ for some integers $a_i$ and $b_i$, where $T$ denotes packet duration. There is no loss of generality in this assumption, because we can approximate any real number by rational numbers with arbitrarily small error.

Let $c$ be the least common multiple of $a_1, a_2, \ldots, a_M$. Then $cT$ is a common period of all protocol signals, because
\[
  s_i(t + cT) = s_i(t + b_i(c/a_i)P_i) = s_i(t)
\]
for all $t$ and for all $i$. In order to simplify notation, it is convenient to define $s_0(t) = s_{M+1}(t) = 0$ for all $t$.

In the slot-synchronous system model, the delay offsets $\delta_1, \delta_2, \ldots, \delta_M$ are fixed  integral multiples of~$T$ in the channel model. For the purpose of our proof, we impose a fictitious probability distribution on the delay offsets, and assume that $\delta_1, \delta_2, \ldots, \delta_M$ are independent random variables, uniformly distributed over $\{0, T, 2T, \ldots, (c-1)T\}$.

Suppose  the duty factor of $s_i(t)$, defined in~\eqref{eq:def_duty_factor}, is equal to $f_i$, for $i=1,2,\ldots, M$.
We want to prove the following claim:  for all $t$, we have
\begin{equation}
 E_{\delta_i}[ s_i(t-\delta_i) ] = f_i, \label{eq:mean_of_s}
\end{equation}
with the expectation taken over $\delta_i$. We note that the left hand side of the above equation is a discrete sum, while the right hand side is an integration. The definition of $f_i$ implies that within a period of $cT$, there are $cf_i $ non-overlapping semi-open intervals of length $T$, in which $s_i(t)$ is equal to~1. Let $t_0$ be a fixed real number between 0 and $cT$, and consider the following set of $c$ time instants
\begin{equation}\{t_0, t_0-T, \ldots, t_0-(c-1)T\}, \label{eq:P_points}
\end{equation}
with the subtraction performed modulo~$cT$.
This is a set of $c$ evenly spaced points in $[0, cT)$.
Because the length of each semi-open interval is $T$, each of these $c f_i$ intervals  contains exactly one time instant in~\eqref{eq:P_points}. If we evaluate $s_i(t)$ at the time instants in~\eqref{eq:P_points}, exactly $cf_i$ of them equal one. At the remaining time instants, the values of $s_i(t)$ are zero. Thus, we get
\[
\sum_{k=0}^{c-1} s_i(t_0 - k T) = cf_i.
\]
Hence,
\[
 E_{\delta_i}[ s_i(t_0-\delta_i) ] = \frac{1}{c}\sum_{k=0}^{c-1} s_i(t_0 - k T) = f_i.
\]
This completes the proof of the claim.

Consider an arbitrary semi-open time interval $[t_0, t_0+cT)$, for some constant $t_0$. Within this semi-open time interval, packets can be sent from node $i$ to node $i+1$ only when node $i$ is in transmit mode, node $i+1$ is in receive mode, and node $i+2$ is not transmitting anything. Let $\mathcal{T}_{i,i+1}$ be a subset of $[t_0, t_0+cT)$ containing time instants that satisfy
\[
\begin{cases}
s_i(t-\delta_i) = 1 \\
s_{i+1}(t-\delta_{i+1}) = 0 \\
s_{i+2}(t-\delta_{i+2}) = 0
\end{cases}
\]
or equivalently,
\begin{equation}
 s_i(t-\delta_i)(1 - s_{i+1}(t -\delta_{i+1})) (1 - s_{i+2}(t -\delta_{i+2})) =1.
 \label{eq:non_collide}
\end{equation}
The set $\mathcal{T}_{i,i+1}$ is the union of some non-overlapping sub-intervals of $[t_0, t_0+cT)$. We define the length of $\mathcal{T}_{i,i+1}$ as the summation of the length of the constituent sub-intervals. Every non-collided packet from node $i$ to node $i+1$ must be transmitted totally within~$\mathcal{T}_{i,i+1}$. Otherwise, it partially overlaps with other packets and is lost due to collision. The total number of non-collided packets from node $i$ to node $i+1$ is thus no larger than the length of $\mathcal{T}_{i,i+1}$ divided by~$T$. We remark that the number of non-collided packets may be strictly less than the length of $\mathcal{T}_{i,i+1}$ divided by~$T$, because we only assume slot-synchronous system, which may not be time-slotted.

By the independence of $\delta_i$, $\delta_{i+1}$ and $\delta_{i+2}$, the expected value of the left hand side of~\eqref{eq:non_collide}, over random variables $\delta_i$, $\delta_{i+1}$ and $\delta_{i+2}$, is
\begin{align*}
&\phantom{=} E \big[  s_i(t-\delta_i)(1 - s_{i+1}(t -\delta_{i+1})) (1 - s_{i+2}(t -\delta_{i+2})) \big] \\
&=E \big[  s_i(t-\delta_i)\big] E\big[(1 - s_{i+1}(t -\delta_{i+1}))\big] \\
& \qquad \cdot E\big[ (1 - s_{i+2}(t -\delta_{i+2})) \big] \\
&= f_i(1-f_{i+1}) (1-f_{i+2}).
\end{align*}
The last equality follows from~\eqref{eq:mean_of_s}.

Let $\mathbb{I}(x)$ be the indicator function,
\[
 \mathbb{I}(x) = \begin{cases}
   1 & \text{if $x$ is true}\\
   0 & \text{otherwise}.
 \end{cases}
\]
The expected length of $\mathcal{T}_{i,i+1}$, taken over random variables $\delta_i$, $\delta_{i+1}$ and $\delta_{i+2}$, is thus
\begin{align*}
 E\Big[ \int_{\mathcal{T}_{i,i+1}} dt \Big]
&= E\Big[ \int_{t_0}^{t_0 + cT} \mathbb{I}(t \in \mathcal{T}_{i,i+1}) \, dt \Big]\\
&=  \int_{t_0}^{t_0 + cT} E\Big[ \mathbb{I}(t \in \mathcal{T}_{i,i+1})  \Big]\, dt \\
&= \int_{t_0}^{t_0 + cT}  f_i(1-f_{i+1}) (1-f_{i+2}) \, dt \\
&= cT f_i(1-f_{i+1}) (1-f_{i+2}).
\end{align*}
We can find some realization of the random variables $\delta_i$, $\delta_{i+1}$ and $\delta_{i+2}$ such that the length of $\mathcal{T}_{i,i+1}$ is less than or equal to the expected value $c T f_i(1-f_{i+1}) (1-f_{i+2})$. Therefore, there are  some specific values of $\delta_i$, $\delta_{i+1}$ and $\delta_{i+2}$ such that
the number of non-collided packets from node $i$ to node $i+1$ in a duration of $cT$ seconds is no more than $c f_i(1-f_{i+1}) (1-f_{i+2})$.

Suppose that there is a friendly genie who in advance informs nodes $i$ and $i+1$  which packets will be collided, and which packets will be received successfully.
Consider a cut of the network by the edge $(i,i+1)$. The decoding of the sources in $\bar{\mathcal{S}}_i^f$ by nodes $i+1$, $i+2, \ldots, M$ is a function of the packets from node $i$ to node~$i+1$.  With the help of the genie, the link from node $i$ to node $i+1$ reduces to a noiseless discrete memoryless channel with packet rate less than or equal to $f_i(1-f_{i+1}) (1-f_{i+2})$ packets per slot duration.
By the max-flow bound for multicast network codes~\cite[Chapter10]{Kramer}~\cite[Chapter 21]{Yeung08}, we deduce that if $(R_1,\ldots, R_N)$ is achievable no matter what the delay offsets are, then
\[
\sum_{j \in \bar{\mathcal{S}}_i^f} R_j \leq f_i(1-f_{i+1}) (1-f_{i+2}).
\]
A rate vector is not achievable with the help of genie is certainly not achievable in the presence of genie.
This proves the rate constraint in~\eqref{eq:upper_bound1}.

The derivation of~\eqref{eq:upper_bound2} for traffic in the backward direction is similar as above,  with $i+1$ and $i+2$ replaced by $i-1$ and $i-2$, respectively.

The above argument holds for any fixed duty factors. The outer bound follows by taking the union over all duty factors.
\end{proof}

\section{Performance Comparison} \label{sec:example}

We compare the capacity region with three random access schemes. The first two schemes do not have the network coding feature, while the third one does.
In all these three schemes, each node maintains two queues, one for incoming packets from the left and one for packets from the right.
We consider the heavy traffic scenario and assume that the queues are saturated for simplified analysis

In the first scheme, each node sends packets  in the fashion of pure ALOHA, and we will call this {\em uncoded pure ALOHA} scheme. Transmission of the nodes are independent from each other. We adopt a simplifying assumption that the transmission times of packets of node $i$ follow a Poisson process with intensity $\lambda_i$ packets per packet time.
We assume that each packet has a header which stores the identity of the sender and packet numbers. Packet number of collided packets are piggybacked to the transmitting node for re-transmission. The overhead due to packet header is neglected. Consider a source node, say node~$i$, that is associated with source~$\sigma$. When a transmission is initiated, node~$i$ transmits a source packet with probability $p^s_i$, a packet to be relayed to the left with probability $p^\ell_i$ and a packet to be relayed to the right with probability $p^r_i$, where $p^s_i$, $p^\ell_i$ and $p^r_i$ are non-negative real numbers such that $p^s_i+p^\ell_i+p^r_i=1$.

We consider the transmission of a source packet  successful if it is received successfully by both node $i-1$ and node~$i+1$.
When a packet is transmitted by node $i$ at time $t_0$, it will be received by node $i-1$ and $i+1$ if nodes $i-2$, $i-1$, $i+1$ and $i+2$ do not transmit any packet in the time interval $[t_0-T, t_0+T]$ of length $2T$.
This leads to the following rate constraint,
\begin{equation}
 R_\sigma  \leq  p^s_i \lambda_i e^{-2 (\lambda_{i-2}+ \lambda_{i-1}+ \lambda_{i+1}+ \lambda_{i+2})}. \label{eq:pALOHA1}
\end{equation}
(We define $\lambda_{-2}$, $\lambda_{-1}$, $\lambda_{M+1}$ and $\lambda_{M+2}$ to be zero.)
In order to forward packets to nodes $i+1$ and $i-1$ with rates  $\sum_{j\in \mathcal{S}_i^f} R_j $ and $\sum_{j\in \mathcal{S}_i^b} R_j$ respectively, it is required that
\begin{align}
 \sum_{j\in \mathcal{S}_i^f} R_j & \leq p^r_i \lambda_i e^{-2(\lambda_{i+1}+\lambda_{i+2})} \label{eq:pALOHA2}\\
 \sum_{j\in \mathcal{S}_i^b} R_j & \leq p^\ell_i \lambda_i e^{-2(\lambda_{i-1}+\lambda_{i-2})}. \label{eq:pALOHA3}
\end{align}
For a node which is not a source node, the rate requirement is the same as~\eqref{eq:pALOHA2} and~\eqref{eq:pALOHA3}, with $p^\ell_i+p^r_i = 1$.

In the second scheme, called {\em uncoded slotted ALOHA}, slot-synchronization is assumed. Node $i$ transmits a packet in a time slot with probability $f_i$. The protocol is similar to the uncoded pure ALOHA scheme. For a source $\sigma$ which is associated with node $i$, we have the following rate requirement,
\begin{equation}
 R_\sigma \leq  p^s_i f_i(1-f_{i-2})(1-f_{i-1})(1-f_{i+1})(1-f_{i+2}) \label{eq:sALOHA1}
\end{equation}
The rate constraints for forward and backward traffic at node $i$ are
\begin{align}
 \sum_{j\in \mathcal{S}_i^f} R_j & \leq p^\ell_i f_i(1-f_{i+1})(1-f_{i+2})  \label{eq:sALOHA2}\\
 \sum_{j\in \mathcal{S}_i^b} R_j & \leq p^r_i f_i(1-f_{i-1})(1-f_{i-2}).  \label{eq:sALOHA3}
\end{align}
(We define $f_{-2}$, $f_{-1}$, $f_{M+1}$ and $f_{M+2}$ as zero.) For a non-source node $i$, we have  two rate constraints as in~\eqref{eq:sALOHA2} and~\eqref{eq:sALOHA3} with $p^\ell_i + p^r_i = 1$.

The third scheme, which is described in~\cite{Sagduyu06}, is a random access scheme with network coding. It is assumed that the system is time-slotted and slot-synchronous, and node $i$ transmits a packet in a time slot with some fixed probability~$f_i$. When a transmission is initiated at node $i$, with probability $p^s_i$ a source packet is transmitted , and  with probability~$1-p^s_i$ the XOR of two packets from opposite direction to be relayed through node $i$ is transmitted. If source $\sigma$ is associated with node $i$, we have
\begin{equation}
 R_\sigma \leq  p^s_i f_i(1-f_{i-2})(1-f_{i-1})(1-f_{i+1})(1-f_{i+2}) , \label{eq:csALOHA1}
\end{equation}
and
\begin{align}
 \sum_{j\in \mathcal{S}_i^f} R_j & \leq (1-p^s_i) f_i(1-f_{i+1})(1-f_{i+2})  \label{eq:csALOHA2}\\
 \sum_{j\in \mathcal{S}_i^b} R_j & \leq (1-p^s_i) f_i(1-f_{i-1})(1-f_{i-2}).  \label{eq:csALOHA3}
\end{align}
For non-source node $i$, the rate constraints are \eqref{eq:csALOHA2} and~\eqref{eq:csALOHA3} with $p_i^s$ set to zero.
We call the third scheme {\em network-coded slotted ALOHA}.

\subsection{Example 1: Two-way Tandem Network}

By Theorem~\ref{thm:capacity}, the capacity region for the two-way tandem network in Example~1 consists of the rate pairs that satisfy
\begin{align}
R_1 &\leq \min_{i=1,2,3} \{ f_i (1-f_{i+1})(1-f_{i+2}) \} \label{eq:twoway_R1} \\
R_2 &\leq \min_{i=2,3,4} \{f_i (1-f_{i-1})(1-f_{i-2})  \} \label{eq:twoway_R2}
\end{align}
for some duty factors $f_1, f_2, f_3, f_4 \in [0,1]$. ($f_0$ and $f_5$ in~\eqref{eq:twoway_R1} and~\eqref{eq:twoway_R2} are set to 0.) The inequality in~\eqref{eq:twoway_R1} is deduced from the requirement that $R_1$ is less than all rate constraints in the forward direction, and~\eqref{eq:twoway_R2} corresponds to the backward direction. The achievable rate region is
plotted in Fig.~\ref{fig:graph}.

We can verify that the point corresponding to $R_1=0$ and $R_2=1/3$ is achievable by putting $f_1 = 0$, $f_2=1/3$, $f_3 = 1/2$ and $f_4=1$ in~\eqref{eq:twoway_R2},
\[
 R_2 = \frac{1}{3}= \min\Big\{ 1\cdot \frac{1}{2} \cdot \frac{2}{3},\
 \frac{1}{2}\cdot \frac{2}{3} \cdot 1, \ \frac{1}{3}\cdot 1 \Big\}.
\]
By setting $f_1=f_2=f_3 = f_4 = 1/3$, we can check that maximal symmetric rate is $4/27$,
\[
R_1=R_2 = \frac{4}{27} = \min\Big\{\frac{1}{3}\cdot\frac{2}{3}\cdot\frac{2}{3},\  \frac{1}{3}\cdot\frac{2}{3}\cdot\frac{2}{3},\  \frac{1}{3}\cdot\frac{2}{3} \Big\} .
\]
Thus the point $(0.1481, 0.1481)= (4/27, 4/27)$ is the maximal symmetric rate pair in the achievable rate region.

%

\begin{figure}
\begin{center}
  \includegraphics[width=3.5in]{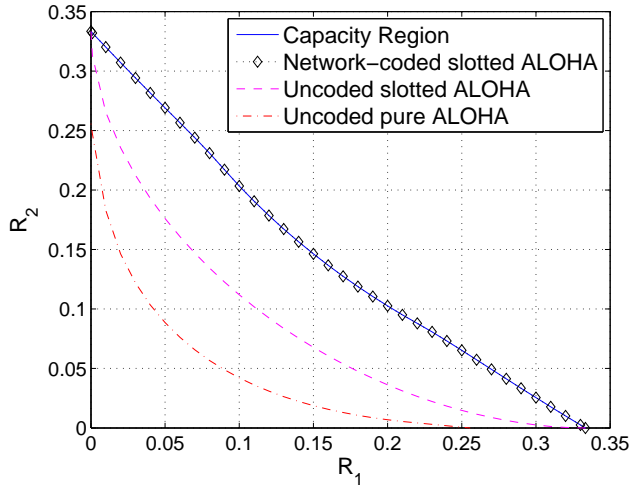}
\end{center}
\caption{The Achievable Rate Region for a Four-Node Two-Way Tandem Collision Network}
\label{fig:graph}
\end{figure}

The achievable rate region by the network-coded slotted ALOHA scheme is {\em identical} to the capacity region defined by~\eqref{eq:twoway_R1} and~\eqref{eq:twoway_R2}. This feature holds for all two-way tandem networks, i.e., only the two end nodes are source nodes, while the nodes in the middle are non-source nodes.

We can also observe that all rate regions are not convex. This is because time-sharing requires inter-node coordination in the multiple-access layer, which is not available in our setting.

The achievable rate regions for uncoded pure ALOHA and uncoded slotted ALOHA are plotted as the dashed lines in Fig.~\ref{fig:graph}.
The symmetric rate achieved by uncoded pure ALOHA and uncoded slotted ALOHA are 0.0678 and 0.1058, respectively, which are much smaller than the maximal symmetric rate 0.1481 achieved by the scheme proposed in this paper and by network-coded slotted ALOHA.

We observe from Fig.~\ref{fig:graph} that the two extreme points  $(0,1/3)$ and $(1/3,0)$ are achievable by all schemes described in this paper except the uncoded pure ALOHA scheme.


\subsection{Example 2: Bi-directional Multicast Network}

In this example, the rate region achieved by network-coded slotted ALOHA is strictly smaller than the capacity region.

By Theorem~\ref{thm:capacity}, the capacity region consists of rate pairs $(R_1,R_2)$ that satisfy
\begin{align*}
R_1 & \leq f_2(1-f_3)(1-f_4) \\
R_1+ R_2 & \leq f_2(1-f_1) \\
R_1 & \leq f_3(1-f_4)(1-f_5) \\
R_2 & \leq f_3(1-f_2)(1-f_1) \\
R_2 & \leq f_4(1-f_3)(1-f_2) \\
R_1+ R_2 & \leq f_4(1-f_5)
\end{align*}
for some duty factors $f_1, f_2, f_3, f_4, f_5 \in [0,1]$. It is easy to see that we can set $f_1=f_5=0$ in the above inequalities without affecting the result. The set of rate constraints that defines the capacity region reduces to
\begin{align}
R_1 & \leq \min\{f_2(1-f_3)(1-f_4), f_3(1-f_4) \} \\
R_1 & \leq \min\{f_4(1-f_3)(1-f_2), f_3(1-f_2) \} \\
R_1+R_2 & \leq \min\{f_2,f_4 \}
\end{align}
where $f_2$, $f_3$ and  $f_4$ are real numbers between zero and one. The capacity region is plotted in Fig.~\ref{fig:graph2}. The maximal symmetric rate is 0.1716.

The reason why the network-coded slotted ALOHA has smaller throughput than our proposed scheme is that the source packets and on-going traffic are treated separately. Nested coding using RS code is capable of encoding packets in an effective way.

As in the previous example, we observe that the two extreme points  $(0,1/3)$ and $(1/3,0)$ are achievable by all schemes except the uncoded pure ALOHA scheme. When two sources are active, uncoded pure ALOHA, uncoded slotted ALOHA and network-coded slotted ALOHA are all suboptimal.\footnote{For the network-coded slotted ALOHA scheme, three different treatments of source packets are presented
in~\cite{Sagduyu06}.  All of them turn out to yield the same achievable rate region in this example.} Unlike the previous example, the network-coded slotted ALOHA scheme is suboptimal in this cases.
The main functionality difference between the network-coded slotted ALOHA scheme and the capacity-achieving scheme is that there is no channel coding in the network-coded slotted ALOHA scheme. Thus, we see that nested coding is essential in improving throughput.

\begin{figure}
\begin{center}
  \includegraphics[width=3.5in]{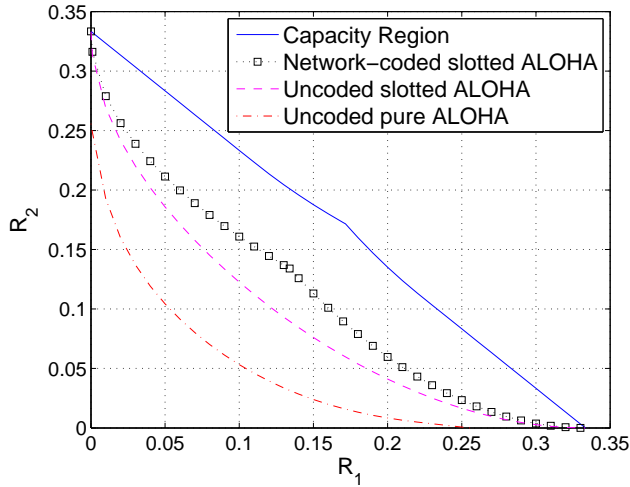}
\end{center}
\caption{The Capacity Regions for the Bi-directional Sensor Network in Example 2}
\label{fig:graph2}
\end{figure}

\section{Conclusion} \label{sec:conclusion}

In this paper we present a transmission scheme for tandem collision network.
The construction involves three different ideas: shift-invariant protocol sequences for multiple-access, nested coding for erasure correction, and network coding for bi-directional information flow. The resulting transmission system does not require any coordination or synchronization among the nodes, and yet is able to achieve optimal throughput in bi-directional tandem collision network. For wireless sensor networks, achieving time synchronization is costly. The proposed scheme can alleviate this problem and operate slot-asynchronously. Numerical examples show that the protocol-signal approach can achieve higher data rate than random-access transmission schemes. A key feature of the proposed transmission that makes it superior to other random-access schemes is erasure-correcting coding. The regularity of channel accessing pattern dictated by protocol sequences make it easy to incorporate channel coding. This feature is not available if the channel accessing is stochastic.

\appendices

\section{Proof of Identifiability} \label{app:identify}
Suppose that the protocol sequences are consecutively 3-wise shift-invariant. Also, suppose $\tau'_{i-1}$ and $\tau'_{i+1}$ are chosen such that the associated channel activity signal $c'[k]$ is the same as $c[k]$ for all $k$.

The number of packets declared to be sent from node $i-1$ is  $P(1-f_i)f_{i-1}(1-f_{i+1})$, because the number of time indices $k$ such that
\[
 c'[k] = s_{i-1}[k-\tau'_{i-1}] = 1
\]
is equal to the throughput $\theta_{i-1,i}( \tau'_{i-1}, \tau_i, \tau'_{i+1})$, which is a constant by the shift-invariant property. Similarly, the number of packets declared to be sent from nodes $i+1$ is equal to
 $P(1-f_i)(1-f_{i-1})f_{i+1}$.
Also, the number of symbol ``*'' in the channel activity signal $c'[k]$ is equals $P(1-f_i)f_{i-1} f_{i+1}$.

In order to show that the sender-identifying algorithm described in Section~\ref{sec:achievable} always yields the correct answer, we consider four cases.

(a) $\tau_{i-1}' = \tau_{i-1}$ and $\tau_{i+1}' = \tau_{i+1}$. The identity of the packets returned by the algorithm is clearly correct in this case.

(b) $\tau_{i-1}' = \tau_{i-1}$ and $\tau_{i+1}' \neq \tau_{i+1}$.
If $c'[k] = 1$ for some $k$, then by construction, either $s_{i-1}[k+\tau_{i-1}] = 1$ or $s_{i+1}[k+\tau'_{i+1}] = 1$
The packet at time $\tilde{k}$ is decided to be sent from node $i-1$ if
\begin{equation}
 c[\tilde{k}] = s_{i-1}[\tilde{k}+\tau_{i-1}] = 1.
 \label{eq:identify1}
\end{equation}
Thus the decisions for the packets that are declared to be from node $i-1$ are all correct. We can find exactly $P(1-f_i)f_{i-1}(1-f_{i+1})$ such $\tilde{k}$ that satisfies~\eqref{eq:identify1} by the shift-invariant property. The remaining $P(1-f_i)(1-f_{i-1})f_{i+1}$ successfully received packets cannot be sent from node $i-1$, and so must be sent from node~$i+1$. Incidently, they are all declared to be sent from node $i+1$, and hence there is no identification error.

(c) $\tau_{i-1}' \neq \tau_{i-1}$ and $\tau_{i+1}' = \tau_{i+1}$. The argument is similar to part (b), with $i-1$ and $i+1$ exchanged.

(d) $\tau_{i-1}' \neq \tau_{i-1}$ and $\tau_{i+1}' \neq \tau_{i+1}$.
If the output of the algorithm was incorrect, then a packet from node $i-1$ would be mistakenly decided to be from node $i+1$, and a packet from node $i+1$ would be mistakenly decided from node~$i-1$. In this case, the number of symbol ``*'' contained in the channel activity signal when the delay offsets of nodes $i-1$ and $i+1$ were $\tau_{i-1}$ and $\tau_{i+1}'$, would be strictly larger than $P(1-f_i)f_{i-1} f_{i+1}$. This contradicts the assumption that the protocol sequences are shift-invariant.




\end{document}